\documentclass{sig-alternate}

\usepackage{amsthm,amssymb,amsmath}
\usepackage{graphics}
\usepackage{tikz}
\usepackage[plainpages=false,pdfpagelabels,colorlinks=true,citecolor=blue,hypertexnames=false]{hyperref}
\usepackage{color}
\newtheorem{theorem}{Theorem}

\newtheorem{prop}[theorem]{Proposition}
\newtheorem{corollary}[theorem]{Corollary}
\newtheorem{lemma}[theorem]{Lemma}

\newtheorem{defi}[theorem]{Definition}
\newtheorem{example}[theorem]{Example}

\def\qed{\quad\rule{1ex}{1ex}}

\def\<#1>{\langle\!\langle#1\rangle\!\rangle} 

\newcommand{\bN}{ {\mathbb N}}

\newcommand{\cS}{ {\mathcal S}}

\let\set\mathbb
\def\lc{\operatorname{lc}}
\def\rank{\operatorname{rank}}
\def\lt{\operatorname{lt}}
\def\im{\operatorname{im}}

\def\diag{\operatorname{diag}}

\begin{document}

\title{Reduction-Based Creative~Telescoping for Algebraic~Functions\titlenote{S.\ Chen was
supported by the NSFC grant 11501552 and
by the President Fund of the Academy of
Mathematics and Systems Science, CAS (2014-cjrwlzx-chshsh). This work was
also supported by the Fields Institute's 2015 Thematic Program on Computer
Algebra in Toronto, Canada.\\
M.\ Kauers was supported by the Austrian Science Fund (FWF): F50-04 and Y464-N18.\\
C.\ Koutschan was supported by the Austrian Science Fund (FWF): W1214.
}}

\numberofauthors{1}

\author{\medskip
Shaoshi Chen$^{1,2}$, \, Manuel Kauers$^{3}$, \, Christoph Koutschan$^{4}$ \\
\smallskip
       \affaddr{$^1$KLMM,\, AMSS, \,Chinese Academy of Sciences, Beijing, 100190, (China)}\\
       \smallskip
       \affaddr{$^2$Symbolic Computation Group, University of Waterloo, Ontario, N2L3G1, (Canada)}\\
              \smallskip
       \affaddr{$^3$Institute for Algebra, Johannes Kepler University, Altenberger Stra\ss e 69,
 A-4040 Linz, (Austria)}\\
        \smallskip
       \affaddr{$^4$RICAM, Austrian Academy of Sciences, Altenberger Stra\ss e 69, A-4040 Linz, (Austria)}\\
       \smallskip
      \email{schen@amss.ac.cn, manuel.kauers@jku.at}\\
      \email{christoph.koutschan@ricam.oeaw.ac.at}
}

\maketitle
\begin{abstract}
Continuing a series of articles in the past few years on creative telescoping using reductions,
we develop a new algorithm to construct minimal telescopers for algebraic functions. This algorithm is
based on Trager's Hermite reduction and on polynomial reduction, which was originally designed for
hyperexponential functions and extended to the algebraic case in this paper.
\end{abstract}

\category{I.1.2}{Computing Methodologies}{Symbolic and Algebraic Manipulation}[Algebraic Algorithms]

\terms{Algorithms, Theory}

\keywords{Algebraic function, Integral basis, Trager's Reduction, Telescoper}


\section{Introduction}\label{SECT:intro}

The classical question in symbolic integration is whether the integral of
a given function can be written in \lq\lq closed form\rq\rq. In its most restricted form,
the question is whether for a given function~$f$ belonging to some domain $D$
there exists another function~$g$, also belonging to~$D$, such that $f=g'$. For
example, if $D$ is the field of rational functions, then for $f=1/x^2$ we can
find $g=-1/x$, while for $f=1/x$ no suitable $g$ exists. When no $g$ exists
in~$D$, there are several other questions we may ask. One possibility is to ask
whether there is some extension~$E$ of $D$ such that in $E$ there exists some
$g$ with $g'=f$. For example, in the case of elementary functions, Liouville's
principle restricts the possible extensions~$E$, and there are algorithms which
construct such extensions whenever possible. Another possibility is
to ask whether for some modification $\tilde f\in D$ of~$f$ there exists a $g\in
D$ such that $\tilde f=g'$. Creative telescoping is a question of this
type. Here we are dealing with domains~$D$ containing functions in several
variables, say $x$ and~$t$, and the question is whether there is a linear
differential operator~$P$, nonzero and free of~$x$, such that there exists a
$g\in D$ with $P\cdot f=g'$, where $g'$ denotes the derivative of $g$ with
respect to~$x$. Typically, $g$~itself has the form $Q\cdot f$ for some operator
$Q$ (which may be zero and need not be free of~$x$). In this case, we call $P$
a telescoper for~$f$, and $Q$ a certificate for~$P$.

Creative telescoping is the backbone of definite integration. Readers not
familiar with this technique are referred to the literature~\cite{PWZbook1996,Zeilberger1990c,Zeilberger1991,Zeilberger1990,Koepf1998}
for motivation, theory, algorithms, implementations, and applications. There are
several ways to find telescopers for a given $f\in D$. In recent years, an
approach has become popular which has the feature that it can find a telescoper
without also constructing the corresponding certificate. This is interesting
because certificates tend to be much larger than telescopers, and in some
applications only the telescoper is of interest. This approach was first
formulated for rational functions $f\in C(t,x)$ in~\cite{BCCL2010} and later
generalized to rational functions in several variables~\cite{bostan13, lairez15}, to
hyperexponential functions~\cite{bostan13a} and, for the shift case, to hypergeometric
terms~\cite{chen15a} and binomial sums~\cite{bostan15}. In the present paper, we will extend
the approach to algebraic functions.

The basic principle of the general approach is as follows. Assume that the
$x$-constants $\mathrm{Const}_x(D)=\{\,c\in D:c'=0\,\}$ form a field and that $D$
is a vector space over the field of $x$-constants. Assume further that there is
some $\mathrm{Const}_x(D)$-linear map $[\cdot]\colon D\to D$ such that for every
$f\in D$ there exists a $g\in D$ with $f-[f]=g'$. Such a map is called a
\emph{reduction.} For example, in $D=C(t,x)$ Hermite reduction~\cite{Hermite1872} produces for
every $f\in D$ some $g\in D$ such that $f-g'$ is either zero or a rational function
with a square-free denominator. In this case, we can take $[f]=f-g'$.
In order to find a telescoper, we can compute $[f]$, $[\partial_t\cdot f]$, $[\partial_t^2\cdot f]$, \dots,
until we find that they are linearly dependent over $\mathrm{Const}_x(D)$.
Once we find a relation
$p_0[f] + \cdots + p_r[\partial_t^r\cdot f] = 0$,
then, by linearity,
$[p_0 f + \cdots + p_r \partial_t^r\cdot f] = 0$,
and then, by definition of $[\cdot]$, there exists a $g\in D$ such that $(p_0+\cdots + p_r\partial_t^r)\cdot f=g'$.
In other words, $P=p_0+\cdots + p_r\partial_t^r$ is a telescoper.

There are two ways to guarantee that this method terminates. The first requires that we already know for
other reasons that a telescoper exists. The idea is then to show that the
reduction $[\cdot]$ has the property that when $f\in D$ is such that there
exists a $g\in D$ with $g'=f$, then $[f]=0$. If this is the case and
$P=p_0+\cdots+p_r\partial_t^r$ is a telescoper for~$f$, then $P\cdot f$ is integrable
in~$D$, so $[P\cdot f]=0$, and by linearity $[f]$, \dots, $[\partial_t^r\cdot f]$ are
linearly dependent over $\mathrm{Const}_x(D)$. This means that the method won't
miss any telescoper. In particular, this argument has the nice feature that we
are guaranteed to find a telescoper of smallest possible order~$r$. This
approach was taken in~\cite{chen15a}.
The second way consists in showing that $\{\,[f]:f\in D\,\}$ is a finite-dimensional vector space over
$\mathrm{Const}_x(D)$. This approach was taken in~\cite{BCCL2010,bostan13a}. It has the
nice additional feature that every bound for the dimension of this vector space
gives rise to a bound for the order of the telescoper. In particular, it implies
the existence of a telescoper.

In this paper, we show that Trager's Hermite reduction for algebraic
functions directly gives rise to a reduction-based creative telescoping
algorithm via the first approach (Section~\ref{SECT:CT-1}). We will combine Trager's Hermite reduction
with a second reduction, called polynomial reduction (Section~\ref{sec:polynomial}), to obtain a reduction-based creative
telescoping algorithm for algebraic functions via the second approach (Section~\ref{SECT:CT-2}).
This gives a new proof of a bound for the order of the telescopers, and in
particular an independent proof for their existence.

A few years ago, Chen et al.~\cite{chen12d} have already considered the problem of creative
telescoping for algebraic functions. They have pointed out that by canceling residues
of the integrand, a given creative telescoping problem can be reduced to a creative
telescoping problem for a function with no residues, which may be much smaller than the
original function. For this smaller function, however, they still need to construct a
certificate. The algorithms presented in the present paper are the first which can find
telescopers for algebraic functions without also constructing corresponding certificates.
By Theorem~6 of~\cite{chen12d}, our results also translate into a certificate-free
creative telescoping algorithm for rational functions in three variables.

\section{Algebraic Functions}

Throughout the paper, let $C$ be a field of characteristic zero, $K=C(t)$, and $\bar K$ the algebraic closure of~$K$.
We consider algebraic functions over~$K$.  For some absolutely irreducible
polynomial $m\in K[x,y]$, we consider the field
$A=K(x)[y]/\langle m\rangle$. If $n=\deg_ym$, then every element of $A$ can be written uniquely in the form
$f=f_0+f_1y+\cdots+f_{n-1}y^{n-1}$ for some $f_0,\dots,f_{n-1}\in K(x)$.

The element $y\in A$ is a solution of the equation $m=0$,
because in $A$ we have $m=0$ by construction. The polynomial $m$ also admits
$n$ distinct solutions in the field
\[
  \bar K\<x-a>:=\bigcup_{r\in\set N\setminus\{0\}} \bar K(\!(\ (x-a)^{1/r}\ )\!)
\]
of formal Puiseux series around $a\in\bar K$. There are also $n$ distinct
solutions in the field
\[
  \bar K\<x^{-1}>:=\bigcup_{r\in\set N\setminus\{0\}} \bar K(\!(x^{-1/r})\!)
\]
of formal Puiseux series around~$\infty$.
Since $\bar K\<x^{-1}>$ and the $\bar K\<x-a>$ are fields, we can associate to every
$f\in A$ and every $a\in\bar K\cup\{\infty\}$ in a natural way $n$ distinct series
objects with fractional exponents, by plugging any of the $n$ distinct series solutions
of $m$ into the representation $f=f_0+\cdots+f_{n-1}y^{n-1}$.
In other words, for every $a\in\bar K\cup\{\infty\}$ there are $n$ distinct natural
ring homomorphisms from $A$ to $K\<x-a>$ or $K\<x^{-1}>$, respectively.

In the field $A$ as well as the fields $\bar K\<x-a>$ and $\bar K\<x^{-1}>$, we have
natural differentiations with respect to~$x$. For a series, differentiation is defined
termwise using the usual rules $\bigl((x-a)^{\nu+n}\bigr)'=(\nu+n)(x-a)^{\nu+n-1}$ and
$\bigl((x^{-1})^{\nu+n}\bigr)'=-(\nu+n)(x^{-1})^{\nu+n+1}$. For the elements of~$A$, note
first that $m(x,y)=0$ implies
\begin{equation}\label{eq:yprime}
  m(x,y)'=(\frac d{dx}m)(x,y) + (\frac d{dy}m)(x,y)y' = 0,
\end{equation}
so $y'=-(\frac d{dx}m)(x,y)/(\frac d{dy}m)(x,y)$. Regarding $m$ as element of $K(x)[y]$
and observing that $0<\deg_y\frac d{dy}m<n$, we have $\gcd(m,\frac d{dy}m)=1$ in $K(x)[y]$,
so $\frac{d}{dy}m$ is invertible in $A=K(x)[y]/\langle m\rangle$.
Note that we have $x'=1$ and $c'=0$ for all $c\in K=C(t)$, in particular also $t'=0$.
The derivative of an arbitrary
element $f\in A$, say $f=p(x,y)$ for some $p\in K(x)[y]$ of degree less than~$n$, is
\[f'=(\frac{d}{dx}p)(x,y)+(\frac d{dy}p)(x,y)y'.\]
Thus we have an action of the algebra $K(x)[\partial_x]$ of differential operators on~$A$.

The derivations on $A$ and on the series domains are compatible in the sense
that for every $f\in A$, the series associated to $f'$ are precisely the
derivatives of the series associated to~$f$.

In the context of creative telescoping, we will also need to differentiate with
respect to~$t$. The action of $K(x)[\partial_x]$ on $A$ and on the series domains
is extended to an action of $K(x)[\partial_x,\partial_t]$ on $A$ and on the series
domains. On $A$, the action of $\partial_t$ is defined as the unique derivation with
$\partial_t\cdot t=1$ and $\partial_t\cdot x=0$, analogously to the construction
above. For the series domains, $\partial_t$ acts on the coefficients
(which are elements of $\bar K$) in the natural way, and does not affect~$x$.
Since each particular element $c\in\bar K$ belongs to a finite algebraic extension of~$K$,
the result $\partial_t\cdot c$ is uniquely determined.
The actions of the larger operator algebra
$K(x)[\partial_x,\partial_t]$ on $A$ and on the series domains are compatible to
each other.

In this paper, the notation $f'$ will always refer to the derivative $\partial_x\cdot f$
with respect to~$x$, not with respect to~$t$.

\medskip

Trager's Hermite reduction for algebraic functions rests on the notion of
integral bases. Let us recall the relevant definitions and properties.
Although the elements of a Puiseux series ring $\bar K\<x-a>$ are formal
objects, the series notation suggests certain analogies with complex
functions. Terms $(x-a)^\alpha$ or $(\tfrac1x)^\alpha$ are
called \emph{integral} if $\alpha\geq0$. A series in
$\bar K\<x-a>$ or $\bar K\<x^{-1}>$ is called integral if it only contains integral
terms. A non-integral series is said to have a \emph{pole} at the reference
point. Note that in this terminology also $1/\sqrt{x}$ has a pole
at~$0$. Note also that the terminology only refers to $x$ but not to~$t$.

Integrality at $a\in\bar K$ is not preserved by differentiation,
but if $f$ is integral at~$a$, then so is $(x-a)f'$. Somewhat conversely,
integrality at infinity is preserved by differentiation, we even have the
stronger property that when $f$ is integral at infinity, then not only $f'$ but also $xf'=(x^{-1})^{-1}f'$ is
integral at infinity.

An element $f\in A=K(x)[y]/\langle m\rangle$
is called (locally) integral at $a\in\bar K\cup\{\infty\}$ if for every series
associated to $y$ the corresponding series for $f$ is integral.
The element $f$ is called (globally) integral if it is locally integral at every
$a\in\bar K$ (``at all finite places'').
This is the case if and only if the minimal polynomial of $f$ in $K[x,y]$ is monic
with respect to~$y$.
Because of Chevalley's theorem~\cite[page 9, Corollary 3]{Chevalley1951}, any
non-constant algebraic function has at least one pole. Equivalently, an element $f$ is
integral at all $a\in\bar K\cup\{\infty\}$ if and only if it is constant.

For an element $f\in A$ to have a ``pole'' at $a\in\bar K\cup\{\infty\}$ means
that $f$ is not locally integral at~$a$; to have a ``double pole'' at $a$ means
that $(x-a)f$ (or $\frac1xf$ if $a=\infty$) is not integral; to have a ``double
root'' at $a$ means that $f/(x-a)^2$ (or $f/(\frac1x)^2=x^2f$ if $a=\infty$) is integral,
and so on.

The set of all globally integral elements $f\in A$ forms a $K[x]$-submodule of~$A$.
A basis $\{\omega_1,\dots,\omega_n\}$ of this module is called an \emph{integral basis}
for~$A$. Such bases exist, and algorithms are known for computing them~\cite{trager84,Rybowicz:1991:ACI:120694.120715,vanHoeij94}.
For a fixed $a\in\bar K$, let $\bar K(x)_a$ be the ring of rational functions $p/q$
with $q(a)\neq0$, and write $\bar K(x)_\infty$ for the ring of all
rational functions $p/q$ with $\deg_x(p)\leq\deg_x(q)$.
Then the set of all $f\in A$ which are locally integral at some
fixed $a\in\bar K\cup\{\infty\}$ forms a $\bar K(x)_a$-module. A basis of this module is
called a \emph{local integral basis} at $a$ for~$A$. Also local integral bases can
be computed.

An integral basis $\{\omega_1,\dots,\omega_n\}$ is always also a $K(x)$-vector space
basis of~$A$. A key feature of integral bases is that they make poles explicit. Writing
an element $f\in A$ as a linear combination $f=\sum_{i=1}^n f_i\omega_i$ for some
$f_i\in K(x)$, we have that $f$ has a pole at $a\in\bar K$ if and only if at least one
of the $f_i$ has a pole there.

\begin{lemma}\label{lemma:1}
  Let $\{\omega_1,\dots,\omega_n\}$ be a local integral basis of $A$ at $a\in\bar K\cup\{\infty\}$.
  Let $f\in A$ and $f_1,\dots,f_n\in K(x)$ be such that $f=\sum_{i=1}^nf_i\omega_i$.
  Then $f$ is integral at $a$ if and only if each $f_i\omega_i$ is integral at~$a$.
\end{lemma}
\begin{proof}
  The direction ``$\Leftarrow$'' is obvious. To show ``$\Rightarrow$'', suppose
  that $f$ is integral at~$a$. Then there exist $w_1,\dots,w_n\in\bar K(x)_a$ such that
  $f=\sum_{i=1}^nw_i\omega_i$. Thus $\sum_{i=1}^n(w_i-f_i)\omega_i=0$, and then
  $w_i=f_i$ for all $i$, because $\omega_1,\dots,\omega_n$ is a vector space basis of~$A$.
  As elements of $\bar K(x)_a$, the $f_i$ are integral at~$a$, and hence also all the $f_i\omega_i$
  are integral at~$a$.
\end{proof}

The lemma says in particular that poles of the $f_i$ in a linear combination
$\sum_{i=1}^n f_i\omega_i$ have no chance to cancel each other.

\begin{lemma}\label{lemma:e}
  Let $\{\omega_1,\dots,\omega_n\}$ be an integral basis of~$A$.
  Let $e\in K[x]$ and
  $M=((m_{i,j}))_{i,j=1}^n\in K[x]^{n\times n}$ be such that
  \[
    e\,\omega_i'=\sum_{j=1}^n m_{i,j}\omega_j
  \]
  for $i=1,\dots,n$ and $\gcd(e,m_{1,1},\dots,m_{n,n})=1$.
  Then $e$ is squarefree.
\end{lemma}
\begin{proof}
  Let $a\in\bar K$ be a root of~$e$. We show that $a$ is not a multiple root.
  Since $\omega_i$ is integral, it is in particular locally integral at~$a$.
  Therefore $(x-a)\omega_i'$ is locally integral at~$a$.
  Since $\omega_1,\dots,\omega_n$ is an integral basis, it follows that
  $(x-a)m_{i,j}/e\in\bar K(x)_a$ for all~$i,j$.
  Because of $\gcd(e,m_{1,1},\dots,m_{n,n})=1$, no factor $x-a$ of $e$
  can be canceled by all the~$m_{i,j}$.
  Therefore the factor $x-a$ can appear in $e$ only once.
\end{proof}

\begin{lemma} \label{lemma:degM}
  Let $\{\omega_1,\dots,\omega_n\}$ be a local integral basis at infinity of~$A$.
  Let $e\in K[x]$ and $M=((m_{i,j}))_{i,j=1}^n\in K[x]^{n\times n}$
  be defined as in Lemma~\ref{lemma:e}. Then $\deg_x(m_{i,j})<\deg_x(e)$ for all $i,j$.
\end{lemma}
\begin{proof}
  Since every $\omega_i$ is locally integral at infinity, so is every $x\,\omega_i'$.
  Since $\omega_1,\dots,\omega_n$ is an integral basis at infinity, it follows that
  $xm_{i,j}/e\in\bar K(x)_\infty$ for all~$i,j$. This means that $1+\deg_x(m_{i,j})\leq\deg_x(e)$
  for all~$i,j$, and therefore $\deg_x(m_{i,j})<\deg_x(e)$, as claimed.
\end{proof}

A $K(x)$-vector space basis $\{\omega_1,\dots,\omega_n\}$ of $A$ is
called \emph{normal} at $a\in\bar K\cup\{\infty\}$ if there exist $r_1,\dots,r_n\in
K(x)$ such that $\{r_1\omega_1,\dots,r_n\omega_n\}$ is a local integral basis
at~$a$. Trager shows how to construct
an integral basis which is normal at infinity from a given integral basis and
a given local integral basis at infinity~\cite{trager84}.

Although normality is a somewhat weaker condition on a basis than integrality,
it also excludes the possibility that poles in the terms of a linear combination
of basis elements can cancel:

\begin{lemma}\label{lemma:3}
  Let $\{\omega_1,\dots,\omega_n\}$ be a basis of~$A$
  which is normal at some $a\in\bar K\cup\{\infty\}$.
  Let $f=\sum_{i=1}^n f_i\omega_i$ for some $f_1,\dots,f_n\in K(x)$.
  Then $f$ has a pole at $a$ if and only if
  there is some $i$ such that $f_i\omega_i$ has a pole at~$a$.
\end{lemma}
\begin{proof}
  Let $r_1,\dots,r_n\in K(x)$ be such that $r_1\omega_1,\dots,r_n\omega_n$ is a
  local integral basis at~$a$. By $f=\sum_{i=1}^n
  (f_ir_i^{-1})(r_i\omega_i)$ and by Lemma~\ref{lemma:1}, $f$~is integral at~$a$ iff all
  $f_ir_i^{-1}r_i\omega_i=f_i\omega_i$ are integral at~$a$.
\end{proof}

\section{Hermite Reduction}\label{sec:hermite}

We now recall the Hermite reduction for algebraic functions~\cite{trager84,ACA1992,bronstein98}.
Let $\{\omega_1,\ldots,\omega_n\}$ be an integral basis for~$A$.
Further let $e, m_{i,j}\in K[x]$ ($1\leq i,j\leq n$) be such that
$e\omega_i'=\sum_{j=1}^n m_{i,j}\omega_i$ and
$\gcd(e,m_{1,1},m_{1,2},\ldots,m_{n,n})=1$.
For describing the Hermite reduction we fix an integrand $f\in A$ and represent it in the
integral basis, i.e., $f=\sum_{i=1}^n (f_i/D)\,\omega_i$ with
$D, f_1,\ldots,f_n\in K[x]$. The purpose is to find $g, h\in A$ such that
$f=g' + h$ and $h=\sum_{i=1}^n(h_i/D^\ast)\,\omega_i$ with $h_1,\ldots,h_n\in K[x]$
and $D^\ast$ denoting the squarefree part of~$D$.
As differentiating the $\omega_i$ can introduce
denominators, name\-ly the factors of~$e$, it is convenient to consider those
denominators from the very beginning on, which means that we shall assume
$e\mid D$. Note that $\gcd(D,f_1,\ldots,f_n)$ can then be nontrivial.
Let~$v\in K[x]$ be a nontrivial squarefree factor of~$D$ of multiplicity~$\mu>1$.
Then~$D = uv^\mu$ for some $u\in K[x]$ with $\gcd(u, v)=1$ and $\gcd(v,v')=1$.
One step of the Hermite reduction is as follows:
\begin{equation}\label{eq:hred}
  \sum_{i=1}^n \frac{f_i}{uv^\mu}\omega_i =
  \biggl(\sum_{i=1}^n\frac{g_i}{v^{\mu-1}}\omega_i\biggr)' +
  \sum_{i=1}^n \frac{h_i}{uv^{\mu-1}}\omega_i,
\end{equation}
where $g_i, h_i \in K[x]$ and~$\deg_x(g_i)< \deg_x(v)$.
The existence of such~$g_i$'s and~$h_i$'s follows from the crucial fact that
the elements $s_i :=  uv^\mu(v^{1-\mu}\omega_i)'$ with $i\in \{ 1, \ldots, n\}$
form a local integral basis at each root of~$v$~\cite[page 46]{trager84}.
By a repeated application of such reduction steps, one can decompose any $f\in A$
as $f=g' + h$ where the denominators of the coefficients of $h$ are squarefree
and the coefficients of $g$ are proper rational functions (i.e., their numerators
have smaller degree than their denominators).

It was observed that Hermite reduction itself often takes less time than the construction
of an integral basis. If Hermite reduction is applied to some other basis, for instance
the standard basis $\{1,y,\dots,y^{n-1}\}$, it either succeeds or it runs into a division by zero.
Bronstein~\cite{bronstein98a} noticed that when a division by zero occurs, then the basis can
be replaced by some other basis that is a little closer to an integral basis, just
as much as is needed to avoid this particular division by zero. After finitely many
such basis changes, the Hermite reduction will come to an end and produce a correct
output. This variant is known as lazy Hermite reduction.

\section{Telescoping via reductions: \hskip0ptplus1fill\break first approach} \label{SECT:CT-1}
Recall from the introduction that reduction-based creative telescoping requires
some $K$-linear map $[\cdot]\colon A\to A$ with the property that
$f-[f]$ is integrable in $A$ for every $f\in A$. This is sufficient for the
correctness of the method, but additional properties are needed in order to
ensure that the method terminates.

As also explained already in the introduction, one possibility consists in
showing that $[f]=0$ whenever $f$ is integrable. Trager showed that his
Hermite reduction has this property~\cite[page 50, Theorem 1]{trager84}.
For the sake of completeness, we reproduce his proof here.

\begin{lemma}\label{lemma:pole_at_inf}
Let $W=\{\omega_1,\dots,\omega_n\}$ be an integral basis for $A$ that is normal at
infinity. Let $g=\sum_{i=1}^ng_i\omega_i\in A$ be such that all its
coefficients $g_i\in K(x)$ are proper rational functions. If an integral
element $f\in A$ has a pole at infinity, then also $f+g$ has a pole at
infinity.
\end{lemma}
\begin{proof}
Since $f$ is assumed to be integral we can write it as
$f=f_1\omega_1+\cdots+f_n\omega_n$ with $f_i\in K[x]$.
If $f$ has a pole at infinity, there is at least one index~$i$
such that $f_i\omega_i$ has a pole at infinity. There are two cases why this
can happen.
\renewcommand{\labelenumi}{(\alph{enumi})}
\begin{enumerate}
\item The polynomial~$f_i$ has positive degree.  This means that $f_i+g_i$ has a
  pole at infinity, because the $g_i$ are proper rational functions.
  Thus $(f_i+g_i)\omega_i$ has a pole at infinity, because $\omega_i$ has no poles
  at finite places and therefore no root at infinity.
\item The integral basis element $\omega_i$ is not constant and $f_i$ is not zero. Hence
  $\omega_i$ has a pole at infinity, and this also implies that $(f_i+g_i)\omega_i$
  has a pole at infinity, again employing the fact that $g_i$ is a proper rational function.
\end{enumerate}
In both cases, therefore, $f+g=\sum_{i=1}^n(f_i+g_i)\omega_i$ has a pole at
infinity by Lemma~\ref{lemma:3}.
\end{proof}

\begin{theorem}\label{thm:intiff0}
  Suppose that $f\in A$
  has a double root at infinity (i.e., every series in $\bar K\<x^{-1}>$
  associated to $f$ only contains monomials $(1/x)^\alpha$ with $\alpha\geq2$).
  Let $W=\{\omega_1,\dots,\omega_n\}$
  be an integral basis for $A$ that is normal at infinity.
  If $f=g'+h$ is the result of the Hermite reduction with respect to~$W$,
  then $h=0$ if and only if $f$ is integrable in~$A$.
\end{theorem}

\begin{proof}
The direction ``$\Rightarrow$'' is trivial. To show the implication
``$\Leftarrow$'' assume that $f$ is integrable in~$A$. From $f=g'+h$ it follows that
then also $h$ is integrable in~$A$; let $H\in A$ be such that $H'=h$.  In order to show
that $h=0$, we show that $H$ is constant.  To this end, it suffices to show that
it has neither finite poles nor a pole at infinity; the claim then follows from
Chevalley's theorem.

It is clear that $H$ has no finite poles because $h$ has at most simple poles
(i.e., all series associated to $h$ have only exponents $\geq-1$).
This follows from the facts that the $\omega_i$ are integral and that
the coefficients of~$h$ have squarefree denominators.

If $H$ has a pole at infinity, then by Lemma~\ref{lemma:pole_at_inf} also
$g+H$ must have a pole at infinity, because Hermite reduction produces
$g=\sum_i g_i\omega_i$ with proper rational functions~$g_i$.  On the other
hand, since $f=g'+h=(g+H)'$ has at least a double root at infinity by
assumption, $g+H$ must have at least a single root at infinity. This is
a contradiction.
\end{proof}

Note that the condition in Theorem~\ref{thm:intiff0} that $f$ has a double
root at infinity is not a restriction at all, as it can always be achieved by
a suitable change of variables. Let $a\in C$ be a regular point; this means
that all series in $\bar K\<x-a>$ associated to $f$ are formal power series. By the substitution
$x\to a+1/x$ the regular point~$a$ is moved to infinity. From
\[
  \int f(x) \,\mathrm{d}x = \int f\left(\frac{1}{x}+a\right)\left(-\frac{1}{x^2}\right) \mathrm{d}x
\]
we see that the new integrand has a double root at infinity.

Moreover, since the action of $\partial_t$ on series domains is defined coefficient-wise,
it follows that when $f$ has at least a double root at infinity (with respect to~$x$),
then this is also true for $\partial_t\cdot f, \partial_t^2\cdot f, \partial_t^3\cdot f,\dots$,
and then also for every $K$-linear combination $p_0f+p_1\partial_t\cdot f+\cdots+p_r\partial_t^r\cdot f$.
Thus Theorem~\ref{thm:intiff0} implies that $p_0+p_1\partial_t+\cdots+p_r\partial_t^r$ is a telescoper for $f$ if
\emph{and only if} $[p_0+p_1\partial_t+\cdots+p_r\partial_t^r]=0$.

We already know for other reasons~\cite{Zeilberger1990,chyzak00,chen12d} that
telescopers for algebraic functions exist, and therefore the re\-duc\-tion-based
creative telescoping procedure with Hermite reduction with respect to an
integral basis that is normal at infinity as reduction function succeeds when
applied to an integrand $f\in A$ that has a double root at infinity.
In particular, the method finds a telescoper of smallest possible order.
Again, if $f$ has no double root at infinity, we can produce one by a change of variables.
Note that a change of variables $x\to a+1/x$ with $a\in C$ has no effect on
the telescoper.

\begin{example}\label{ex:ct}
We consider the algebraic function $f=y/x^2$ where $y$ is a solution of the
third-degree polynomial equation $m(x,y) = y^3 + y + x + t = 0$. An integral
basis for $A=K(x)[y]/\langle m\rangle$ that is normal at infinity is given by
$\omega_1=1$, $\omega_2=y$, $\omega_3=y^2$.  (This means that employing lazy
Hermite reduction avoids completely the computation of an integral basis in
this example.)

By solving Equation~\eqref{eq:yprime} for $y'$ we obtain
\[
  y' = \frac{-6y^2 + 9(t+x)y - 4}{27x^2+54tx+27t^2+4}.
\]
Then for the differentiation matrix~$\frac1eM$, a simple calculation yields
\[
  \begin{pmatrix} \omega_1' \\[1pt] \omega_2' \\[1pt] \omega_3' \end{pmatrix} =
  \frac{1}{e} \begin{pmatrix} 0 & 0 & 0 \\[1pt] -4 & 9 (t+x) & -6 \\[1pt] 12 (t+x) & 4 & 18 (t+x) \\ \end{pmatrix}
  \begin{pmatrix} \omega_1 \\[1pt] \omega_2 \\[1pt] \omega_3 \end{pmatrix}
\]
with $e=27x^2+54xt+27t^2+4$. Thus we write $f=\sum_{i=1}^3 (f_i/D) \omega_i$ with
$f_1=f_3=0$, $f_2=e$, and $D=x^2e$.  After a single step the Hermite reduction
delivers the result
\[
  f = \biggl(\, \underbrace{\vphantom{\frac{1}{(x)}} -\frac{y}{x}}_{\textstyle=g_0} \;\biggr)' +\>
  \underbrace{\frac{-6y^2+9(x+t)y-4}{x(27x^2+54xt+27t^2+4)}}_{\textstyle=h_0}.
\]
As the Hermite remainder~$h_0$ is nonzero, Theorem~\ref{thm:intiff0} tells us that
$f$ is not integrable in~$A$. Hence we continue by applying Hermite reduction to
\[
  \partial_t\cdot f = \frac{-6y^2+9(x+t)y-4}{x^2(27x^2+54xt+27t^2+4)}.
\]
Note that we could as well take $\partial_t\cdot h_0$ instead of $\partial_t\cdot f$, which
in general should result in a faster algorithm.
Again after a single reduction step, the decomposition $\partial_t\cdot f = g_1' + h_1$
is obtained, where
\begin{align*}
  g_1 &= \frac{6y^2-9ty+4}{x(27t^2+4)} \\
  h_1 &= \frac{6\bigl((9x+27t)y^2-(27xt+27t^2-2)y+6x+18t\bigr)}{x(27t^2+4)(27x^2+54xt+27t^2+4)}.
\end{align*}
Since $h_0$ and $h_1$ are linearly independent over $K=C(t)$, we continue with
$\partial_t^2\cdot f$.
This time however, it is preferable to start the Hermite reduction
with $\partial_t\cdot h_1$, which is given by
\[
  \frac{1}{x(27t^2+4)^2(27x^2+54xt+27t^2+4)^2}.
\]
Setting $v=27x^2+54xt+27t^2+4=e$ and doing one reduction step,
the Hermite remainder $h_2$ is found to be
\begin{multline*}
 \bigl(6\bigl((-729xt-1539t^2+96)y^2+(1215xt^2-144x+1215t^3-{}\\
 306t)y-486xt-1026t^2+64\bigr)\bigr)\mathrel{\big\slash}\bigl(x(27t^2+4)^2e\bigr).
\end{multline*}
The corresponding integrable part $g_2$ is not displayed here for space reasons.

Now one can find a linear dependence between $h_0,h_1,h_2$ that gives rise to the telescoper
$(27t^2+4)\partial_t^2+81t\partial_t+24$, which is indeed the minimal one for this example.
\end{example}

\section{Polynomial Reduction}\label{sec:polynomial}

Recall that instead of requesting that $[f]=0$ if and only if $f$ is integrable
(first approach), we can also justify the termination of reduction-based
creative telescoping by showing that the $K$-vector space $\{\,[f]:f\in A\,\}$
has finite dimension (second approach). If $[\cdot]$ is just the Hermite
reduction, we do not have this property. We therefore introduce below an
additional reduction, called \emph{polynomial reduction,} which we apply after
Hermite reduction. We then show that the combined reduction (Hermite reduction
followed by polynomial reduction) has the desired dimension property for the
space of remainders. As a result, we obtain a new bound on the order of the
telescoper, which is similar to those in~\cite{chen12d,chen14a}.

In this approach, we use two integral bases. First we use a global integral basis (not
necessarily normal at infinity) in order to perform Hermite reduction. Then we write the
remainder $h$ with respect to some local integral basis at infinity and perform the
polynomial reduction on this representation.

Throughout this section let $W=(\omega_1,\ldots,\omega_n)^T\in A^n$ be such
that $\{\omega_1, \ldots, \omega_n\}$ is a global integral basis of~$A$, and
let $e\in K[x]$ and $M=(m_{i,j})\in K[x]^{n\times n}$ be such that $eW'=MW$
and $\gcd(e, m_{1, 1}, m_{1, 2}, \ldots, m_{n ,n})=1$. The Hermite reduction
described in Section~\ref{sec:hermite} decomposes an input element $f\in A$
into the form
\[
  f = g' + h = g' + \sum_{i=1}^n \frac{h_i}{de} \omega_i,\qquad
  g, h\in A,
\]
with $h_i, d\in K[x]$ such that $\gcd(d, e)=\gcd(h_i,de)=1$ and $d$ is squarefree.
\begin{lemma}\label{LEM:d}
If $h$ is integrable in~$A$, then $d$ is in~$K$.
\end{lemma}
\begin{proof}
Suppose that $h$ is integrable in~$A$, i.e., there exist $a, b_i\in K[x]$
such that $h = \bigl(\frac{1}{a}\sum_{i=1}^n b_i \omega_i\bigr)'$. Then
\[
  h = \sum_{i=1}^n \frac{h_i}{de}\omega_i= \sum_{i=1}^n \biggl(\Bigl(\frac{b_i}{a}\Bigr)' \omega_i +
  \frac{b_i}{a e} \sum_{j=1}^n  m_{i, j}\omega_j\biggr).
\]
We show that $a$ is constant. Otherwise, for any irreducible factor $p$ of~$a$, we would have that $h$ has a pole of
multiplicity greater than $1$ at the roots of~$p$. This contradicts
the fact that $d, e$ are squarefree. Thus, $d$ is a constant.
\end{proof}

By the extended Euclidean algorithm, we compute $u_i, v_i\in K[x]$ such that
$h_i = u_i d + v_i e$ and $\deg_x(v_i) < \deg_x(d)$. Then the Hermite remainder~$h$
decomposes as
\begin{equation}\label{EQ:h}
  \sum_{i=1}^n \frac{h_i}{de}\omega_i =  \sum_{i=1}^n \frac{u_i}{e}\omega_i + \sum_{i=1}^n \frac{v_i}{d}\omega_i.
\end{equation}

We now introduce the \emph{polynomial reduction} whose goal is to confine the $u_i$ to a finite-dimensional
vector space over~$K$. Similar reductions have been introduced and used in creative telescoping
for hyperexponential functions~\cite{bostan13a} and hypergeometric terms~\cite{chen15a}.
Let $V = (\nu_1, \ldots, \nu_n)^T\in A^n$ be such that its entries form a $K(x)$-basis of~$A$,
and let $a\in K[x]$ and $B = (b_{i, j})\in K[x]^{n \times n}$ be such that $aV'=BV$ and
$\gcd(a, b_{1, 1}, b_{1, 2}, \ldots, b_{n ,n})=1$. Let $p = (p_1, \ldots, p_n)\in K[x]^n$. Then
\begin{equation} \label{EQ:polyred}
  (pV)' = \sum_{i=1}^n (p_i \nu_i)' = \frac{ap' + pB}{a}\, V.
\end{equation}
This motivates us to introduce the following definition.
\begin{defi}
Let the map $\phi_V\colon K[x]^n \rightarrow K[x]^n$
be defined by $\phi_V(p) = ap' + pB$ for any $p\in K[x]^n$.
We call $\phi_V$ the \emph{map for polynomial reduction} with respect to~$V$, and call
the subspace $\im(\phi_V) = \{\phi_V(p) \mid p \in K[x]^n\}$
the \emph{subspace for polynomial reduction} with respect to~$V$.
\end{defi}

Note that, by construction and because of Lemma~\ref{LEM:d}, $q\in K[x]^n$ is in
$\im(\phi_V)$ if and only if $\frac{q}{a}V$ is integrable in~$A$.

We can always view an element of $K[x]^n$ (resp. $K[x]^{n\times n}$) as a polynomial in~$x$
with coefficients in~$K^n$ (resp. $K^{n\times n}$). In this sense we use the notation $\lc(\cdot)$
for the leading coefficient and $\lt(\cdot)$ for the leading term of a vector (resp. matrix).
For example, if $p\in K[x]^n$ is of the form
\[
  p = p^{(r)}x^r + \dots + p^{(1)}x + p^{(0)},\quad p^{(i)}\in K^n,
\]
then $\deg_x(p)=r$, $\lc(p)=p^{(r)}$, and $\lt(p)=p^{(r)}x^r$.
Let $\{e_1, \ldots, e_n\}$ be the standard basis of~$K^n$.
Then the module $K[x]^n$ viewed as a $K$-vector space is generated by
\[
  \cS := \bigl\{e_ix^j \mathrel{\big|} 1\leq i \leq n,\, j\in \bN\bigr\}.
\]
We define $K[x]_\mu^n:=\{p\in K[x]^n \mid \deg_x(p) \leq \mu\}$; as a $K$-vector
space it is generated by
\[
  \cS_\mu := \bigl\{e_ix^j \mathrel{\big|} 1\leq i \leq n,\, 0\leq j\leq \mu\bigr\}.
\]
Any element $p\in K[x]_\mu^n$ can be expressed in the
basis $\cS_\mu$ as a vector $\vec{p}\in K^{n(\mu+1)}$ (in the following the
decoration~\raisebox{-1pt}{$\vec{}\;$} always indicates such a typecast).

\begin{defi}
Let $N_V$ be the $K$-subspace of $K[x]^n$ generated by
\[
  \bigl\{t \in \cS \mathrel{\big|} t \neq \lt(p) \ \text{for all $p\in \im(\phi_V)$}\bigr\}.
\]
Then $K[x]^n = \im(\phi_V) \oplus N_V$.
We call $N_V$ the \emph{standard complement} of $\im(\phi_V)$.
For any $p\in K[x]^n$, there exist $p_1\in K[x]^n$ and~$p_2\in N_V$ such that
\[\frac{p}{a}V = (p_1V)' + \frac{p_2}{a}V.\]
This decomposition is called the \emph{polynomial reduction} of~$p$
with respect to~$V$.
\end{defi}

\begin{prop}\label{PROP:finite}
Let $a\in K[x]$ and $B\in K[x]^{n \times n}$ be such that $aV'=BV$, as before.
If $\deg_x(B) \leq \deg_x(a)-1$, then $N_V$ is a finite-dimensional
$K$-vector space.
\end{prop}
\begin{proof}
In addition to the proof of the assertion, we also explain how to determine
the dimension and a basis for $N_V$, for later use. For brevity, let
$\mu:=\deg_x(a)-1$. We distinguish two cases.

\smallskip
{\it Case 1.}~
Assume that $\deg_x(B) < \mu$. For any $p\in K[x]^n$ of degree $s>0$, we have
\[
  \lt\bigl(\phi_V(p)\bigr) = s\lc(a)\lc(p)x^{s+\mu}.
\]
Thus all monomials $e_i x^j\in \cS$ with $1\leq i\leq n$ and $j\geq \mu+1$ are not in~$N_V$.
Let $\vec{B}_1, \ldots, \vec{B}_n$ be the columns of~$B$, expressed in the basis $\cS_\mu$.
Let $C(B)$ be the $K$-subspace of $K[x]_\mu^n$ generated by these column vectors.
If $q\in \im(\phi_V)$, then $q = \phi_V(p) = pB$ for some $p \in K^n$, which implies that
$\vec{q}\,$ is a linear combination of $\vec{B}_i$'s. Then $K[x]_\mu^n = C(B) \oplus N_V$.
So $\dim_K(N_V)= (\mu+1)n - \dim_K(C(B))$ and a basis of $N_V$ can be computed by
looking at the echelon form of the matrix $\bigl(\vec{B}_1, \ldots, \vec{B}_n\bigr)$.

\smallskip
{\it Case 2.}~
Assume that $\deg_x(B) =\mu$. For any $p\in K[x]^n$ of degree $s$, we have
\[
  \lt\bigl(\phi_V(p)\bigr) = \lc(p)(s\lc(a)I_n + \lc(B))x^{s+\mu}.
\]
Let $\ell$ be the largest nonnegative integer such that $-\ell \lc(a)$ is an
eigenvalue of $\lc(B)\in K^{n\times n}$. Then for any $s>\ell$,
the matrix $J_s = s\lc(a)I_n + \lc(B)$ is invertible. So any monomial $e_ix^j$ with $j> \ell+\mu$ is not in~$N_V$
for any $i=1, \ldots, n$. Let $p = \sum_{i=1}^n \sum_{j=0}^{\ell} p_{i, j} e_ix^j$.
Then $\phi_V(p)$ belongs to $K[x]_{\ell+\mu}^n$.
In the basis $\cS_{\ell+\mu}$, we can
express $\phi_V(p)$ as a vector of length ${n(\ell+\mu+1)}$ with entries linear in the $p_{i, j}$'s.
This vector can be written in the form $M_{\ell} \vec{P}$,
where $\vec{P} = (p_{1, 0}, p_{2, 0}, \ldots, p_{n, \ell})^T$ and $M_{\ell} \in K^{n(\ell+\mu +1) \times n(\ell+1)}$.
Every $q\in K[x]_{\ell+\mu}^n$ can be expressed as a vector $\vec{q} \in K^{n(\ell + \mu +1)}$.
Then $q\in\im(\phi_v)$ if and only if $\vec{q}\,$ is in the column space of~$M_{\ell}$.
Therefore,
\[K[x]_{\ell+\mu}^n = C({M_{\ell}}) \oplus N_V. \]
This implies that $\dim_K(N_V) = n(\ell+\mu+1) - \rank({M_{\ell}})$, and
a basis of $N_V$ can be computed by
looking at the echelon form of the matrix ${M_{\ell}}$.
\end{proof}

In general, the condition $\deg_x(B) \leq \deg_x(a)-1$ may not hold for an arbitrary basis~$V$ of~$A$.
The following lemma shows that we can perform a simple change of basis to make the condition hold.

\begin{lemma}\label{LM:CB}
Let~$W =\{\omega_1, \ldots, \omega_n\}$ be an integral basis of~$A$ such that it is also normal at infinity. Then
there exist nonnegative integers~$\tau_1, \ldots, \tau_n$ such that
\[ V := \{\nu_1, \ldots, \nu_n\} \quad \text{with $\nu_i = x^{-\tau_i} \omega_i$}\]
is a basis of~$A$ which is normal at $0$ and integral at all other places (including infinity).
\end{lemma}
\begin{proof}
It is clear that such a basis~$V$ will be normal at zero, because multiplying the generators by
the rational functions $x^{\tau_i}$ brings it back to a global integral basis, which is in particular
a local integral basis at zero.
It is also clear that such a basis will be integral at every other point $a\in\bar K\setminus\{0\}$, because the
multipliers $x^{-\tau_i}$ are locally units at such~$a$.
Finally, since the original basis is normal at infinity, there exist rational functions $u_1,\dots,u_n$
such that $\{u_1\omega_1,\dots,u_n\omega_n\}$ is a local integral basis at infinity.
Since $u_i$ can be written as $u_i=x^{-\tau_i}\tilde{u}_i$ with $\tau_i\in\set Z$ and $\tilde{u}_i$ being a unit
in $\bar{C}(x)_\infty$, we see that also $V$ is a local integral basis at infinity.
The integers~$\tau_i$ can only be nonnegative because the $\omega_i$'s have no finite poles and therefore each
of them is either constant or has a pole at infinity by Chevalley's theorem.
\end{proof}

Combining the Hermite reduction and polynomial reduction, we get the following theorem.
\begin{theorem}\label{THM:polyred}
Let $W$ be an integral basis of~$A$ that is normal at infinity.
Let $T := \diag(x^{-\tau_1}, \ldots, x^{-\tau_n}) \in K(x)^{n\times n}$
be such that $V = TW$ is integral at infinity.
Let $e\in K[x]$, $\lambda \in \bN$, and $B, M \in K[x]^{n \times n} $ be such that
$eW' = MW$ and $x^\lambda eV' = BV$.
Then any element $f\in A$ can be decomposed into
\begin{equation}\label{EQ:add}
f = g' + \frac{1}{d} PW + \frac{1}{x^\lambda e} QV,
\end{equation}
where $g\in A$, $d\in K[x]$ is squarefree and $\gcd(d, e)=1$, $P, Q\in K[x]^n$ with $\deg_x(P) < \deg_x(d)$ and $Q\in N_V$, which is
a finite-dimensional $K$-vector space. Moreover, $P, Q$ are zero if and only if $f$ is integrable in~$A$.
\end{theorem}
\begin{proof}
After performing the Hermite reduction on $f$, we get
\[f = \tilde{g}' + \frac{1}{d} PW + \frac{1}{e} UW,\]
where~$P = (v_1, \ldots, v_n)\in K[x]^n$ and $U = (u_1, \ldots, u_n)\in K[x]^n$
with~$u_i, v_i$ introduced in~\eqref{EQ:h}. By Lemma~\ref{LM:CB}, there exists
$T := \diag(x^{-\tau_1}, \ldots, x^{-\tau_n}) \in K(x)^{n\times n}$
such that $V = TW$ is normal at $0$ and integral at any other places (including infinity). Note that we can
choose~$T$ as the identity matrix if~$\deg_x(M)\leq \deg_x(e)-1$.
By taking derivatives, we get
\[V' = \left(T' + T\frac{M}{e}\right)T^{-1}V = \frac{B}{a}V, \]
where $a=x^\lambda e$ for some $\lambda\in \bN$ and $B\in K[x]^{n\times n}$. Since $V$ is locally integral
at infinity, $\deg_x(B) \leq \deg_x(a)-1$ by Lemma~\ref{lemma:degM}.
By expanding in terms of the new basis~$V$, we get
\[\frac{1}{e} UW = \frac{1}{a} \tilde{U}V, \]
where $\tilde{U} = x^\lambda U T^{-1} \in K[x]^n$. Next, we decompose $\tilde{U}$ into
$\tilde{U} = \phi_{V}(\tilde{U}_1) + \tilde{U}_2$ with $\tilde{U}_1, \tilde{U}_2\in K[x]^n$ and
$\tilde{U}_2\in N_V$. Then we get
\[\frac{1}{e} UW = (\tilde U_1 V)' + \frac{1}{a} \tilde U_2 V. \]
We then get the decomposition~\eqref{EQ:add} by setting
$g = \tilde g + \tilde U_1 V$ and~$Q = U_2$.

Assume that $f$ is integrable. Then Lemma~\ref{LEM:d} implies that $d\in K$.
Since $\deg_x(P) < \deg_x(d)$, we have $P=0$. Then
\[\frac{1}{x^\lambda e} QV = \sum_{i=1}^n (a_i \nu_i)'\]
for some $a_i\in K[x]$. So $Q \in \im(\phi_V)$.
Since $\im(\phi_V) \cap N_V = \{0\}$, it follows that $Q=0$.
\end{proof}
The decomposition in~\eqref{EQ:add} is called an \emph{additive decomposition} of~$f$ with respect to~$x$.

\section{Telescoping via reductions: \hskip0ptplus1fill\break second approach} \label{SECT:CT-2}

We now discuss how to compute telescopers for algebraic functions via Hermite reduction and
polynomial reduction.

Let $W, V, e, \lambda, M, B$ be as in Theorem~\ref{THM:polyred}.
To construct a telescoper for~$f\in A$,
we first consider the additive decompositions of the successive derivatives $\partial_t^i\cdot f$ for $i\in \bN$.
Assume that
\[\partial_t\cdot W = \frac{1}{\tilde{e}} \tilde{M}W \quad \text{and}
\quad \partial_t\cdot V = \frac{1}{x^{\tilde{\lambda}} \tilde{e}} \tilde{B}V,\]
where $\tilde e \in K[x]$, $\tilde M, \tilde B\in K[x]^{n\times n}$, and $\tilde \lambda \in \bN$.
Since $\partial_t$ and~$\partial_x$ commute, Proposition 7 in~\cite{chen14a}
implies that $\tilde e \mid e$ and $x^{\tilde{\lambda}} \tilde{e} \mid x^\lambda e$, as polynomials in $K[x]$.
So we can just take $\tilde e = e$ and $\tilde{\lambda}  = \lambda$ by multiplying $\tilde M, \tilde B$ by some factors
of~$x^\lambda e$. A direct calculation yields $\partial_t\cdot f = (\partial_t\cdot g)' + h$,
where
\[h = \left(\partial_t\cdot\frac{P}{d}+\frac{P\tilde M}{de}\right)W + \left(\partial_t\cdot\frac{Q}{x^\lambda e}+ \frac{Q\tilde B}{x^{2\lambda} e^2}\right)V.\]
This implies that the squarefree part of the denominator of $h$ divides $xde$. Applying Hermite reduction and polynomial reduction
to~$h$ yields
\[ h = \tilde g_1' + \frac{1}{d} P_1W + \frac{1}{x^\lambda e} Q_1V,\]
where $P_1, Q_1\in K[x]^n$ with $\deg_x(P_1) < \deg_x(d)$ and $Q_1\in N_V$.
Repeating this discussion, we get the following lemma.
\begin{lemma}\label{LEM:idtf}
For any $i\in \bN$, the derivative $\partial_t^i\cdot f$ has an additive decomposition of the form
\[ \partial_t^i\cdot f = g_i' + \frac{1}{d} P_iW + \frac{1}{x^\lambda e} Q_iV,\]
where $g_i\in A$, $P_i, Q_i\in K[x]^n$ with $\deg_x(P_i) < \deg_x(d)$ and $Q_i\in N_V$.
\end{lemma}
As application of the above lemma, we can compute the minimal telescoper for $f$ by finding the first
linear dependence among the $(P_i, Q_i)$ over~$K$. We also obtain an upper bound for the order of telescopers.
\begin{corollary}
Every $f\in A$ has a telescoper of order at most $n\deg_x(d) + \dim_K(N_V)$.
\end{corollary}

\begin{example}
We continue with Example~\ref{ex:ct}, by applying the polynomial reduction
to the Hermite remainders $h_0,h_1,h_2$. The matrix~$M$ computed before
satisfies the degree condition of Proposition~\ref{PROP:finite}, so no
change of basis is needed. First we compute polynomials
$u_i,v_i\in K[x,y]$ such that for $i=0,1,2$ we have
\[
  h_i = \frac{u_i}{e} + \frac{v_i}{d} = \frac{u_i}{27x^2+54xt+27t^2+4} + \frac{v_i}{x}.
\]
By noting that $\deg_x(u_i)=1$ and $\deg_x(e)=2$, we see that the
map for the polynomial reduction $\phi(p) = ep' + pM$ can only be
applied for $p\in K^n$ so that it turns into $\phi(p) = pM$.
This means that we reduce $xy^2$ using the third row of~$M$ and
$xy$ using its second row. A straightforward calculation reveals
that $h_0$, $h_1$, and $h_2$ all reduce to~$0$. Hence we are left
with finding a $K$-linear combination among the $v_i$:
\begin{align*}
 v_0 &= \frac{-6y^2+9yt-4}{27t^2+4},\\
 v_1 &= \frac{6\bigl(27y^2t-(27t^2-2)y+18t\bigr)}{(27t^2+4)^2},\\
 v_2 &= \frac{6\bigl((96-1539t^2)y^2+(1215t^3-306t)y-1026t^2+64\bigr)}{(27t^2+4)^3}.
\end{align*}
As expected, we obtain the same telescoper as in Example~\ref{ex:ct}.
\end{example}

\section{The D-finite Case}

With algebraic functions being settled, it is natural to wonder about a possible
reduction-based creative telescoping algorithm for D-finite functions. Recall
that in this setting we consider an operator $L\in K(x)[\partial_x]$ instead of a minimal
polynomial $m\in K[x,y]$ and instead of an algebraic field extension
$K(x)[y]/\langle m\rangle$ we consider the $K(x)[\partial_x]$-left-module
$A=K(x)[\partial_x]/\langle L\rangle$. Then the element $1\in A$ is a solution of $L$
because $L\cdot 1=L=0$ in $A$ by construction. If $n=\deg_{\partial_x}L$, then
the general element of $A$ has the form
$f=f_0+f_1\partial_x+\cdots+f_{n-1}\partial_x^{n-1}$ for some
$f_0,\dots,f_{n-1}\in K(x)$. Very much as in the algebraic case, there is a natural way
to associate certain series objects to the elements of~$A$. Based on these
series objects, a notion of integrality was proposed last year~\cite{kauers15b}, and an
algorithm for computing integral bases has been given for so-called Fuchsian
operators~$L$.

It turns out that the Hermite reduction of Section~\ref{sec:hermite} also works in this setting, if we say that
a term $(x-a)^\alpha\log(x)^\beta$ in a generalized series solution is integral if and only if $\alpha\geq0$.
Note that then $\log(x)$ will then be considered integral at zero, despite the singularity of the complex
function at this point. This has the somewhat counterintuitive consequence that $\log(x)$ is integral at
every $a\in\bar K\cup\{\infty\}$ although it does not have a pole anywhere. For algebraic functions,
this is not possible by Chevalley's theorem, and this fact enters in an essential way in the proofs of
Sections~\ref{SECT:CT-1} and~\ref{sec:polynomial}. The lack of Chevalley's theorem is not an artefact of a (possibly
wrong) treatment of logarithmic terms. Because of the Fuchs relation \cite[p.~241]{schlesinger95} there exist operators
$L\in K(x)[\partial_x]$ whose series solutions at any point $a\in\bar K\cup\{\infty\}$ have no logarithmic
terms, only nonnegative exponents, and which are nevertheless not constant.

For the time being, the existence of such operators is a severe obstruction to a possible generalization of
the termination arguments for reduction-based creative telescoping from algebraic functions to Fuchsian D-finite
functions. We hope to explore this topic further in the future.

\section*{Acknowledgements}

We would like to thank Ruyong Feng and Michael F.\ Singer for helpful discussions.

\bibliographystyle{abbrv}
\bibliography{Hermite}

\end{document}